\documentclass[aps,twocolumn,pra,superscriptaddress,nofootinbib]{revtex4-2}

\pdfoutput=1
\usepackage{microtype}
\usepackage{amsmath}
\usepackage{amssymb}
\usepackage{amsthm}
\usepackage{mathrsfs}
\usepackage{bm, dsfont}
\usepackage[usenames,dvipsnames]{color}
\usepackage{colortbl}
\usepackage[table]{xcolor}
\usepackage{enumitem}
\usepackage{booktabs}
\usepackage{natbib}
\usepackage[colorlinks=true,linkcolor=blue,citecolor=blue,urlcolor=blue]{hyperref}
\usepackage{cleveref}
\usepackage{hypcap}
\usepackage{verbatim, float}
\usepackage{psfrag}
\usepackage[normalem]{ulem}
\usepackage{physics}
\usepackage{listings}
\usepackage[demo]{graphicx}
\usepackage{tikz}
\usepackage{diagbox}	
\usepackage{relsize, scalerel}

\usepackage{subeqnarray}
\usepackage{geometry}
\newcommand{\cH}{\mathcal{H}}

\newcommand{\be}{\begin{equation}}
\newcommand{\ee}{\end{equation}}

\newcommand{\cE}{\mathcal {E}}

\definecolor{EffGreen}{RGB}{34,139,90}

\definecolor{midnightblue}{rgb}{0.1378,0.3784,0.4838}

\hypersetup{
	colorlinks,
	linkcolor=midnightblue,
	citecolor=midnightblue,
	filecolor=midnightblue,
	urlcolor=midnightblue,
}

\newtheorem{theorem}{Theorem}

\newtheorem{observation}[theorem]{Observation}

\newtheorem*{remark}{Remark}

\begin{document}

\title{Quantifying randomness with measurement incompatibility}

\author{Sebastian Schlösser}\thanks{These authors contributed equally.}
\affiliation{
Department of Physics and Astronomy, Uppsala University, Box 516, 751 20 Uppsala, Sweden}
\affiliation{
Nordita, KTH Royal Institute of Technology and Stockholm University,
Hannes Alfv\'ens v\"ag 12, 10691 Stockholm, Sweden}

\author{Pauli Jokinen}\thanks{These authors contributed equally.}
\affiliation{
Department of Physics and Astronomy, Uppsala University, Box 516, 751 20 Uppsala, Sweden}
\affiliation{
Nordita, KTH Royal Institute of Technology and Stockholm University,
Hannes Alfv\'ens v\"ag 12, 10691 Stockholm, Sweden}

\author{Martin Plávala}\thanks{These authors contributed equally.}
\affiliation{Institut f\"ur Theoretische Physik, Leibniz Universit\"at Hannover, Hannover, 30167, Germany}

\author{Leevi Leppäjärvi}
\affiliation{Faculty of Information Technology, University of Jyväskylä, Finland}

\author{Leonardo S. V. Santos}
\affiliation{
Nordita, KTH Royal Institute of Technology and Stockholm University,
Hannes Alfv\'ens v\"ag 12, 10691 Stockholm, Sweden}

\author{Roope Uola}
\affiliation{
Department of Physics and Astronomy, Uppsala University, Box 516, 751 20 Uppsala, Sweden}
\affiliation{
Nordita, KTH Royal Institute of Technology and Stockholm University,
Hannes Alfv\'ens v\"ag 12, 10691 Stockholm, Sweden}

\begin{abstract}
We present a trade-off between the amount of observed measurement incompatibility and the capabilities of a classical Eavesdropper in a prepare-and-measure scenario. The result is based on a qualitative connection between measurement incompatibility and randomness generation together with the utilization of incompatibility witnesses as randomness certificates. This allows one to use a geometric measure of incompatibility, the generalised robustness, to bound Eve’s strategies through a semi-definite program, while providing an explicit protocol for generating randomness from any set of incompatible measurements. By translating the result to quantum steering, we find a tight connection between steerability and randomness generation in a setting using any finite number of measurement inputs. We further show how our techniques can be generalised to scenarios where Eve has a quantum memory by using a dimensional generalisation of joint measurability.
\end{abstract}

\maketitle

\section{Introduction}

Incompatibility of measurements is one of the building blocks of quantum theory. From uncertainty relations~\cite{coles17} to fundamental precision bounds~\cite{busch13,buschrmp2014}, quantum measurements play a crucial role throughout quantum foundations~\cite{busch2016quantum,heinosaari16b}. More recently, measurement incompatibility has found interest in quantum information theory~\cite{JMReviewGuhne} through various connections to, e.g., quantum correlations~\cite{wolf09,marco2014JM,uola2014JM,uola2015JM,tavakoli2020JM,plavala2024IncmultiBell,ducuara2025,Porto2026} and discrimination tasks~\cite{carmeli19a,skrzypczyk19,uola19b,oszmaniec19,uola19c,buscemi20,kuramochi20}. This has further lead to indirect links between incompatibility and practical quantum-technological tasks, such as benchmarking of photon sources \cite{designolle21} and quantum information carriers~\cite{Engineer2025}, as well as data-analysis methods for quantum correlation tests~\cite{Tavakoli2025bintheory}.

In this theory work, we report a direct use-case for incompatibility in the practical task of random number generation. We do this by converting measurement incompatibility witnesses into randomness generation protocols in semi device-independent prepare-and-measure scenarios. 
In such scenarios, Alice sends quantum messages to a receiver, Bob, who has a collection of untrusted (black-box) measurements. 
Their task is to generate randomness based on the observed data without an external party, Eve, having access to it.

To reach randomness, we first note that the quantum information remaining after Eve's attack is characterised by measurements that are pairwise compatible with Eve's strategy~\cite{heinosaari15b,pello7,uola2022retrievability}. Second, with the help of measurement-theoretic tools, we show that any incompatibility in the receiver's end is a guarantee of randomness. For any given set of incompatible measurements, we then construct a randomness certification protocol and establish a monotonous dependency between randomness and incompatibility. More precisely, we establish a trade-off between Eve's possible strategies and the witnessed incompatibility robustness~\cite{haapasalo15a,uola2015JM}, a well-established measure of measurement incompatibility. This bound can be computed efficiently with a single run of a semidefinite program.

\begin{figure}
    \centering
    \includegraphics[width=1\linewidth]{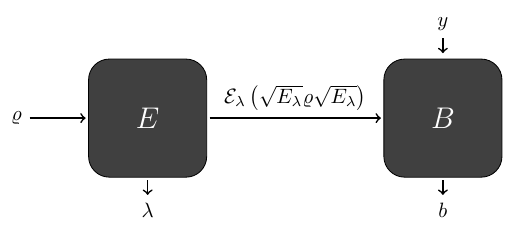}
    \caption{ 
    A trusted quantum state enters a black box device ($E$) from which Eve obtains information $\lambda$. Eve can use channels $\cE_\lambda$ as a decoy in order not to be detected. The quantum message enters Bob's black-box device ($B$), the input $y$ of which is not controlled by Eve.}
    \label{fig:PM_with_eve}
\end{figure}

A key advantage of our approach is that it relies exclusively on measurements, where any degree of incompatibility is linked to randomness. This contrasts with previous results that are restricted to scenarios with two inputs or measurement structures such as star-incompatibility~\cite{passaro15,minagawa2026,Li2024}.
Translating our results to steering allows one to build such connection for any measurement structure and, hence, to improve upon the noise tolerance of state-of-the-art steering-based randomness generation protocols. Finally, we show that dimensional simulability \cite{Ioannou2022} of measurements can be used to model a more general Eve with a qubit memory.

\newpage

\section{Prepare-and-measure scenario with Eve}
We are interested in prepare-and-measure settings: Alice prepares a collection of trusted quantum states acting onto a finite-dimensional Hilbert space, and sends them to Bob, whose measurements act as black-boxes with inputs $y$ and outputs $b$. We allow for an eavesdropper (Eve) in the middle, who can intercept the message and send some disturbed version of it to Bob, see Fig.~\ref{fig:PM_with_eve}.

We assume a classical Eve, meaning that she does not have a quantum memory. With trusted input states, the most general operation she can perform is described by an instrument, a collection of completely positive (CP) maps $\{\mathcal I_\lambda\}$ summing to a trace-preserving (TP) map, i.e.~a quantum channel. Any instrument is of the form $\varrho\mapsto \cE_\lambda(\sqrt{E_\lambda}\varrho\sqrt{E_\lambda})$\footnote{See, e.g., Theorem 7.2 of Ref.~\cite{HayashiBook06} or Corollary 1 of Ref.~\cite{Pello4}.}. Here $\{E_\lambda\}$ extracts classical information for Eve and is described by a positive-operator-valued measure (POVM), $E_\lambda\geq0$ (positive semidefinite) and $\sum_\lambda E_\lambda=\openone$ (the identity operator). Now $\{\cE_\lambda\}$ are CPTP maps from the input Hilbert space to Bob, whose Hilbert space is uncharacterised. Eve uses these channels as a decoy in order not to be detected.

The task of Alice and Bob is to bound Eve's knowledge about Bob's outcomes: the less she knows, the more ``random'' the outcome is. Bob's input may become available to Eve after she has performed her measurement. In such scenario, the randomness quantifier reads~\cite{konig2009operational}
\begin{subeqnarray}\label{eq:H_min}
H_{\rm min}(B\mid E, Y) &=& -\log P_g \qquad \text{with}\\
P_g &=& \sum_y p(y) {\rm Pr}(e_y(\lambda) = b). 
\end{subeqnarray} 
Here, Eve can post-process her outcome $\lambda$ into a guess $e_y(\lambda)$ of Bob's outcome upon learning Bob's input $y$. The inputs are sampled with probability $p(y)$, which one can further optimise over.

\section{Measurement incompatibility}

A crucial notion in our study is measurement incompatibility. A collection of POVMs $\{B_{b|y}\}$ is called jointly measurable (JM) if there exists a single POVM $\{G_\lambda\}$ and probability distributions $\{p(\cdot|y,\lambda)\}$ such that
\begin{equation}
B_{b|y}=\sum_\lambda p(b|y,\lambda)G_\lambda.
\end{equation}
A collection of POVMs that is not JM is called incompatible. As we will see, JM can be used to characterise Eve's capabilities in our setting.

After Eve has performed her measurement, Bob measures the system using some uncharacterised POVMs $\{\tilde B_{b|y}\}$ acting onto a possibly infinite-dimensional Hilbert space. The non-restricted dimension reflects the fact that Bob's measurements are untrusted. Their joint statistics takes the form
\begin{equation}
p_{EB}(b,\lambda\mid y)=\text{tr}\big[\varrho\,\sqrt{E_\lambda}\cE_\lambda^*(\tilde B_{b|y})\sqrt{E_\lambda}\big].
\end{equation}
Here, we have expressed Eve's instrument in the Heisenberg picture\footnote{Note that, while $\lambda$ could be an extra input to Bob's black box, given that it is independent of $y$, it can be absorbed into the channel $\cE_\lambda$, so we choose to use the simpler notation.}. By identifying
\begin{equation}\label{eq:Glambdaby}
G_{\lambda,b|y}:=\sqrt{E_\lambda}\cE_\lambda^*(\tilde B_{b|y})\sqrt{E_\lambda},
\end{equation}
we have a collection of POVMs that characterises the statistics of the prepare-and-measure scenario. Marginalising over Bob's outcomes $b$ yields Eve's measurement $\{E_\lambda\}$, whereas marginalising over $\lambda$ yields \textit{Bob's effective POVMs} $\{B_{b|y}\}$.

By construction, for every $y$, Bob's effective POVM $B_{b|y}$ is JM with Eve's POVM $\{E_\lambda\}$. 
Moreover, the opposite also holds true~\cite{heinosaari15b,pello7,uola2022retrievability}: the statistics of any jointly measurable pair of measurements can be obtained from a sequential scenario. We formalise this remark in the following, and, for completeness, provide its proof in \hyperref[app:A]{Appendix A}.

\begin{remark}
    In a prepare-and-measure scenario, a classical Eve is fully characterised by those POVMs that are pairwise jointly measurable with all Bob's effective measurements.
\end{remark}

We conclude that Eve's knowledge can be limited by obtaining some information about Bob's effective POVMs $\{B_{b|y}\}$. It is crucial to note that whereas Bob's measurements $\{\tilde B_{b|y}\}$ are fully uncharacterised, Bob's effective POVMs $\{B_{b|y}\}$ act onto Alice's Hilbert space which has characterised input states. This makes it possible to probe, for example, their incompatibility properties. Intuitively, the more incompatible Bob's effective measurements, the fewer strategies Eve can have. We make this statement more precise in the following sections.

\section{Security from incompatibility}

We are ready to prove a tight connection between incompatibility and security. We write Eve's guessing probability as
\begin{equation}
P_g=\sum_{y,e}p(y)p_{EB}(e,e|y),
\end{equation}
which includes post-processing of $\lambda$ after $y$ becomes available to Eve. We call the scenario secure, if $P_g<1$. Without loss of generality, we can limit Eve's post-processings to deterministic ones by defining
\begin{equation}
G_{\mathbf{e},b|y}:=\sum_\lambda\prod_{\tilde{y}} p(e_{\tilde{y}}|\lambda,\tilde{y})G_{\lambda,b|y},
\end{equation}
where $\mathbf{e}=(e_1,...,e_n)$ is a vector of Eve's outcomes with $e_{\tilde{y}}$ corresponding to her guess for Bob's measurement $\tilde{y}$. The POVM $G_{\mathbf{e},b|y}$ is a pairwise joint measurement for $B_{b|y}$ and Eve's measurement $E_{\mathbf{e}}:=\sum_b G_{\mathbf{e},b|y}$. As such, it can be realised in a sequence where Eve first measures $E_{\mathbf{e}}$ and Bob then measures his black-box measurement with input $y$~\cite{heinosaari15b,pello7,uola2022retrievability}. In the following Observation, we note that Eve being able to guess perfectly the outcomes of Bob's measurement $y$ is a very strong condition, in that it is equivalent to the $y$th marginal of $E_{\mathbf{e}}$  being equal to Bob's effective POVM $B_{b|y}$.

\begin{observation}\label{Obs:security}
    A prepare-and-measure randomness generation scenario based on full-rank input states is secure iff Bob's effective POVMs are incompatible.
\end{observation}
\begin{proof}[Sketch of proof]
We first note that JM is closely related to Eve's possibility to guess Bob's outcomes. If Bob's effective POVMs $\{B_{b|y}\}$ are JM, one can represent them using a (minimal) Naimark dilation of a marginal form joint measurement $\{G_\lambda\}$ as $B_{b|y}=J^*P_{b|y}J$ with mutually commuting projections $\{P_{b|y}\}$. This allows Eve to use the following strategy: She first applies the Naimark isometry $J$ to the input states and then measures the dilated joint measurement. After this, Bob's measurements are simply the projections $\{P_{b|y}\}$ on the dilation space, which commute with Eve's measurement. Therefore, \emph{Eve can record Bob's outcomes without disturbing the system.}

Conversely, we show that if Eve can guess outputs of all of Bob's effective POVMs, then they are JM. Here we assume $p(y)\neq0$ for all $y$; for other scenarios see \hyperref[app:B]{Appendix B}. For each $y$ we have $1=\sum_e p_{EB}(e,e|y)$. For a specific input $y^*$ this implies 
\begin{equation}
0=\text{tr}\Big[\varrho\,\Big(B_{e|y^*}-
\sum_{\{\mathbf{e}|e_{y^*}=e\}} G_{\mathbf{e},e|y^*}\Big)\Big].
\end{equation}
By noting that $\sum_{\mathbf{e},e_{y^*}=e} G_{\mathbf{e},e|y^*}\leq\sum_{\mathbf{e}} G_{\mathbf{e},e|y^*}= B_{e|y^*}$ and further assuming that the state $\varrho$ has full rank, one sees that 
\begin{subeqnarray}
B_{e|y^*}&=&\sum_{\{\mathbf{e}|e_{y^*}=\ e\}} G_{\mathbf{e},e|y^*}\\
&=&\sum_{\{\mathbf{e}|e_{y^*}=\ e\}}\sum_b G_{\mathbf{e},b|y^*}\\
&=&\sum_{\{\mathbf{e}|e_{y^*}=\ e\}} E_\mathbf{e}.
\end{subeqnarray}
Here, we have used $\sum_b B_{b|y^*}=\openone$ on the second line and $\sum_b G_{\mathbf{e},b|y^*}=E_{\mathbf{e}}$ on the third. As the reasoning and Eve's measurement are independent of $y^*$, the claim follows.
\end{proof}

We note that the assumption of a full-rank state is crucial, as there are incompatible measurements whose restriction to any subspace is jointly measurable~\cite{Loulidi2021,uola21}. We further note that \hyperref[Obs:security]{Observation 1} is a proof-of-principle statement: it provides a connection between incompatibility and security, but not a recipe for randomness generation. On the one hand, it uses only one input state per measurement input $y$, which is insufficient for deciding whether generic effective measurements of Bob are incompatible. On the other hand, it does not state how much Eve's guessing probability deviates from $1$. In the following, we will see that sending only one state per measurement input can lead to scenarios where the guessing probability approaches one asymptotically as the amount of incompatibility decreases. We will demonstrate how to overcome this by probing more of the space, i.e., by sending more states.

\section{Quantification of incompatibility and randomness}

In order to quantify incompatibility, we use the so-called generalised robustness measure.
It measures how much a collection of measurements can be mixed with any other collection before becoming jointly measurable. Formally, the incompatibility robustness of POVMs $\{B_{b|y}\}$ is defined as
\begin{align}\label{IR}
         IR(B_{b|y}) = & \min t\ge 0 \\
        & \text{s.t. } \frac{B_{b|y} + t N_{b|y}}{1+t} \in \text{JM} \nonumber
\end{align}
where the optimisation is over all collections of POVMs $\{N_{b|y}\}$ and JM stands for the jointly measurable set. This optimisation problem can be cast as a semi-definite program~\cite{uola2015JM}, and is closely related to state-discrimination tasks \cite{carmeli19a,skrzypczyk19,uola19b,oszmaniec19,uola19c,buscemi20,kuramochi20}.

The robustness measure can be lower-bounded by using trusted positive semi-definite operators $\{\varrho_{b|y}\}$ that form a so-called incompatibility witness. These fulfill
\begin{equation}
\sum_{b,y}\text{tr}[\varrho_{b|y}O_{b|y}]\leq \beta_{\text{JM}},
\end{equation}
for all jointly measurable collections of measurements $\{O_{b|y}\}$. By the duality theory of semi-definite programs, a violation of this bound by a collection of measurements $\{B_{b|y}\}$ puts a lower bound on their incompatibility robustness, i.e.
\begin{equation}
\sum_{b,y}\text{tr}[\varrho_{b|y}B_{b|y}]\leq(1+IR(B_{b|y}))\beta_{\text{JM}}.
\end{equation}
As only the input states are trusted, these lower bounds are semi-device independent in that the measurements can be treated as black boxes, which fits together with our description of Bob's effective measurements.

We split incompatibility witnesses into two scenarios: security and randomness rounds. The idea is that the label $y'$ of the sent state and the label $y$ of the measurement are drawn pseudo-randomly at the beginning of each round. In this way there will be rounds with $y'=y$, which give a lower bound $\alpha$ on incompatibility robustness, and rounds with $y\neq y'$, which are used to generate randomness.
 
To incorporate this choice, we pick states $\varrho_{b^*|y+1}$ for randomness generation, in which case Eve's guessing probability is given by the following SDP:
\begin{align}\label{Eq:EveGuess2MUB}
        P_g =  \max &\sum_{y,e}p(y)\text{tr}\left[\varrho_{b^*|y+1}\sum_{\{\mathbf{e}|e_{y}=e\}} G_{\mathbf{e},e|y}\right] \\
        \text{s.t. } &\sum_{b,y}\text{tr}\left[ \varrho_{b|y}\sum_{\mathbf{e}} G_{\mathbf{e},b|y}\right]\geq (1+\alpha)\beta_{\text{JM}}\nonumber\\
        &G_{\mathbf{e},b|y}\geq0, \;\;\sum_{\mathbf{e},b}G_{\mathbf{e},b|y}=\openone\nonumber\nonumber\\
        &\sum_b G_{\mathbf{e},b|y}=\sum_b G_{\mathbf{e},b|y'}.\nonumber
\end{align}
We note that one can vary $p(y)$ and $\varrho_{b^*|y+1}$, but as we shall see in the following examples, choosing $p(y)$ to be uniform and the states $\varrho_{b^*|y+1}$ from the used incompatibility witness is sufficient for certifying randomness in basic examples.

\begin{figure*}
    \centering
    (a)
    \includegraphics[width=0.45\linewidth]{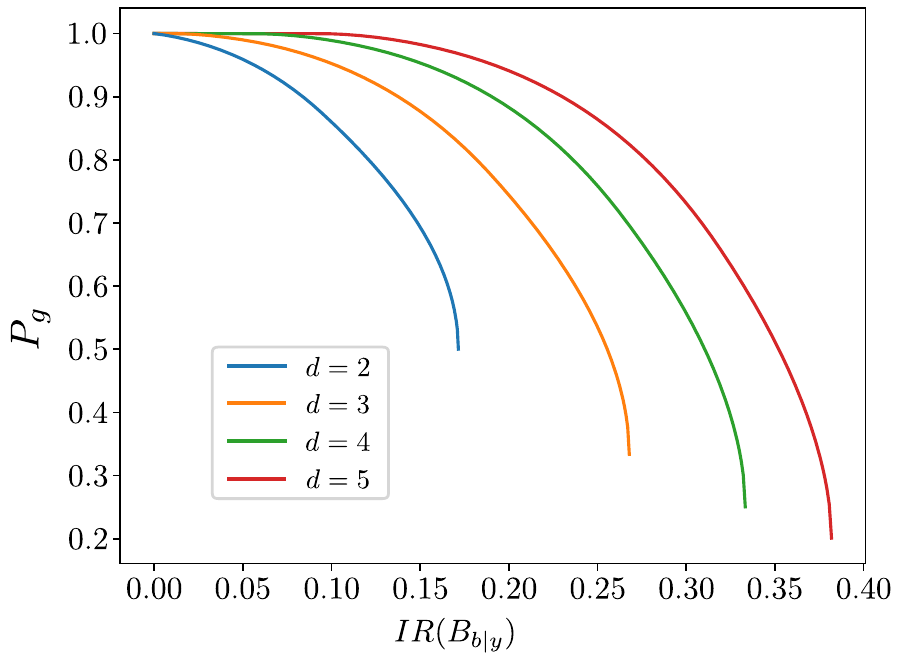}
    (b)
    \includegraphics[width=0.45\linewidth]{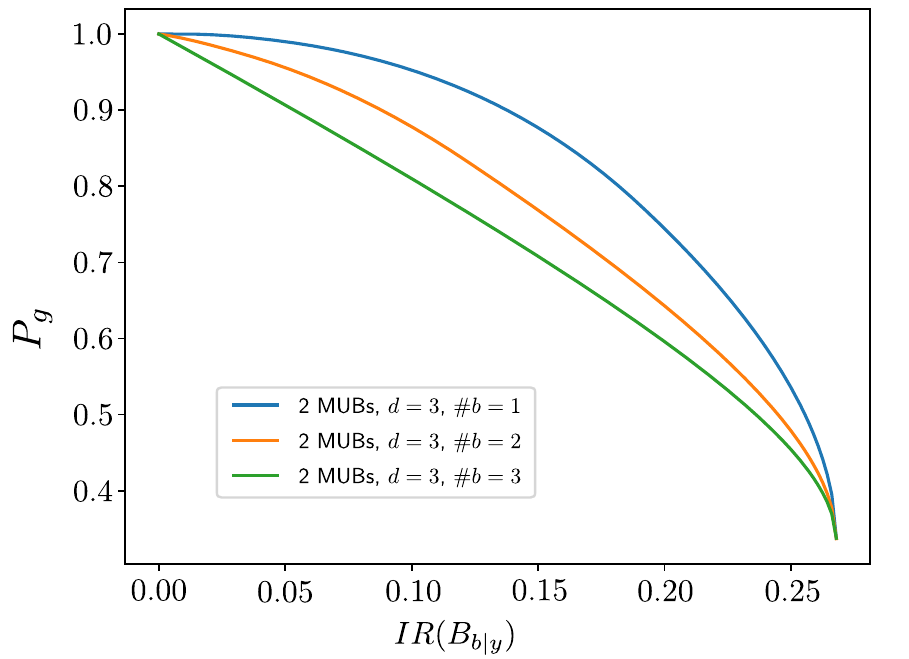}
    \caption{{Eve's guessing probability as a function of the observed incompatibility robustness [Eq.~\eqref{Eq:EveGuess2MUB}].} (a) Two mutually unbiased bases in dimensions $2$ to $5$. Since each of these sets of PVMs have different maximal incompatibility robustness, the lines stop at different points. (b) Two mutually unbiased bases in dimension 3 with multiple input states per basis.}
    \label{fig:IR-Pg-two}
\end{figure*}

\subsection*{Example 1: Two MUBs}

The optimisation problem in Eq.~\eqref{Eq:EveGuess2MUB} suggests a quantitative trade-off between Eve's capabilities and the amount of incompatibility possessed by Bob's effective measurements. To demonstrate this, we take two $d$-outcome measurements for Bob and an incompatibility witness based on two mutually unbiased bases (MUBs) $\varrho_{b|y}=|\varphi_{b|y}\rangle\langle\varphi_{b|y}|$\footnote{Although our security proof in \hyperref[Obs:security]{Observation 1} requires faithful states, by adding any non-zero amount of white noise to the witness allows one to use the above inequality, although the detection strength may slightly drop.}. This witness lower-bounds the incompatibility robustness of 
$B_{b|y}$ as~\cite{designolle21} 
\begin{align}\label{Eq:WitnessBound}
    \frac{1}{d+\sqrt{d}}\sum_{b,y}\text{tr}[B_{b|y}\varrho_{b|y}]\leq1+IR(B_{b|y}).
\end{align}
We exemplify the use of this witness for quantifying randomness in Fig.~\ref{fig:IR-Pg-two}(a).

\subsection*{Example 2: Dealing with the plateaus}

To deal with the asymptotic behaviour in Fig.~\ref{fig:IR-Pg-two}(a), one can add more input states for each $y$. This causes two modifications. First, the guessing probability now includes a sum over the additional states, i.e. over the label $b$. Second, as the information about the label $b$ can become public, Eve can base her post-processing on this information. Thus the guessing probability needs to be updated to $\sum_{b,y,e}p(b,y)\text{tr}\left[\varrho_{b|y+1}\sum_{\{\mathbf{e}|e_{b,y}=e\}} G_{\mathbf{e},e|y}\right]$. Here Eve's vector of outcomes is updated to have a component for each pair $(b,y)$. In Fig.~\ref{fig:IR-Pg-two}(b) we demonstrate the behaviour using the two MUB witness in a qutrit, showing that the asymptotic behaviour changes when using more input states.

\subsection*{Example 3: Randomness from more measurements}

In the case of multiple measurements, one can generate randomness in a noisier setting than when using only two inputs, yet the analysis has to be done more carefully. 

Let us demonstrate this with an example. Consider three inputs for Bob. The triplet of Bob's effective measurements may be incompatible while every pair is JM~\cite{heinosaarifop2008}. An example is that of a hollow triangle: three MUBs in a qubit with visibility in the interval $(1/\sqrt{3},1/\sqrt{2}]$. Due to Eve being pairwise compatible with all Bob's effective measurements, such structure may allow Eve to perfectly guess the outcomes of one of the measurements, but she nevertheless cannot know the outcomes of the other two: if Eve could guess the outcomes of, say, the first and the second of Bob's effective measurements, by \hyperref[Obs:security]{Observation 1} her POVM would be a joint measurement of these two, thus implying triple-wise JM, a contradiction.

This suggests that randomness can be generated from incompatibility even when some pairwise (or $n$-wise) compatibility relations are present. 
The key point is that one needs to generate randomness from various measurements.
In fact, for the hollow triangle, all pairs of measurements are compatible~\cite{designolle19b} and, hence, randomness cannot be extracted from a single measurement.
By contrast, one can still witness incompatibility of the triplet and use the obtained value $\alpha$ as an input for the SDP in Eq.~\eqref{Eq:EveGuess2MUB}. For dimensions $2$, $3$ and $4$, the witnessed robustnesses we found, respectively, $0.082$, $0.107$ and $0.125$, and the corresponding guessing probabilities $0.951$ ,$0.980$ and $0.995$ when using a single state per basis. When using all states of the witness, the qubit SDP runs easily on a laptop and gives $0.888$ for guessing probability.

\subsection*{Example 4: Eve with a quantum memory.} 

Let us consider an Eve, who has access to a single-shot quantum memory. Her strategy is to perform her initial measurement $\{E_\lambda\}$ by applying an instrument $\{\mathcal I_\lambda\}$ that has two quantum outputs: Eve's quantum memory and Bob's black box system. We allow Eve to base her measurement in the memory system on Bob's input $y$. This leads to an effective POVM $G_{\lambda,\lambda',b|y}$, where $\lambda$ comes from Eve's initial measurement, $\lambda'$ from Eve's memory measurement and $b$ from Bob's measurement. The POVM has the explicit form $G_{\lambda,\lambda',b|y}:=\mathcal I_\lambda^*(M_{\lambda'|y,\lambda}\otimes\tilde{B}_{b|y})$. In the following, we consider Eve with a qubit memory for simplicity and note that the theoretical calculations generalise directly to high-dimensional single-shot memories, while the numerical methods will require more computational resources.

We first note that if Eve with a qubit memory can perfectly guess Bob's outcome for some faithful input states, we can use the same argumentation as in Observation \ref{Obs:security} to prove that Bob's effective measurements need to be qubit simulable. This means that $B_{b|y}=\sum_\mu\tilde{\mathcal{I}}_\mu^*(N_{b|y,\mu})$ for some instrument $\tilde{\mathcal{I}}_\mu$ with Kraus operators having rank less than or equal to two \cite{Ioannou2022}. We present the proof of this claim for completeness in Appendix~\ref{app:memory-toSimulability}.

Based on incompatibility robustness of Bob's effective POVMs, one can limit the possible strategies of Eve. We show in Appendix~\ref{app:memory-fromSimulability} that Eve's guessing probability can be bounded by the following SDP hierarchy:

\begin{align}\label{Eq:EveQuantumGuess}
       P_g \leq  \max &\sum_{y,e}p(y)\text{tr}\left[\varrho_{b^*|y+1}\sum_{e} G_{e,e|y}\right] \\
        \text{s.t. } &\sum_{b,y}\text{tr}\left[ \varrho_{b|y}\sum_{e} G_{e,b|y}\right]\geq (1+\alpha)\beta_{\text{JM}}\nonumber\\
        &G_{e,b|y}\geq0, \;\;\sum_{e,b}G_{e,b|y}=\openone\nonumber\nonumber\\
        &\{E_{e|y} = \sum_{b}G_{e,b|y}\}\quad \text{is qubit simulable,}\nonumber
\end{align}
where we note that since Eve's POVMs can depend on $y$ explicitly, we can reduce the vector of Eve's outcomes $\mathbf{e}$ to a single label $e$ from the same outcome set as $b$.

The above hierarchy can be relaxed to a single SDP by requiring that Eve's measurements do not violate a given incompatibility witness by a higher amount that what is allowed by qubits. Such bounds are known for two MUBs \cite{jones2023equivalence,designolle21}. In the case of qutrit input states, the upper bound in Eq.~(\ref{Eq:WitnessBound}) becomes $2\sqrt{2} / (1+\sqrt{2})$. With this, we can certify a guessing probability of 0.924 for Eve when incompatibility robustness is maximal.

\section{Quantum steering}

Incompatibility is closely related to the task quantum steering~\cite{uola20review}. In this task, one focuses on state assemblages, that is collections of operators of the form
\begin{equation}
\sigma_{b|y}=\text{tr}_B[(\openone\otimes B_{b|y})\varrho_{AB}].
\end{equation}
Here $\varrho_{AB}$ is a quantum state shared with a trusted Alice and untrusted Bob. A state assemblage is called unsteerable if
\begin{equation}
\sigma_{b|y}=\sum_\mu p(b|y,\mu)\sigma_\mu.
\end{equation}
for some positive semi-definite operators $\sigma_\mu$. 

In this context, a canonical scenario for randomness extraction is one in which Eve has her own system and tries to guess Bob's output. This is mathematically described with state assemblages
\begin{equation}
\sigma_{e,b|y}=\text{tr}_{BE}[(\openone\otimes B_{b|y}\otimes E_e)\varrho_{ABE}].
\end{equation}
and the guessing probability is given as
\begin{equation}
P_g=\sum_e\text{tr}[\sigma_{e,e|y}].
\end{equation}
To align this with our scenario, we first replace Eve's outcome $e$ with the vector $\bf{e}$. We note that the crucial elements in our analysis are the POVMs $G_{\mathbf{e},b|y}$, which are related to state assemblages via $\sigma^{1/2}G_{\mathbf{e},b|y}\sigma^{1/2}=\sigma_{\mathbf{e},b|y}$, for some full-rank state $\sigma$. By the SGHJW~\cite{gisin89,hughston93} theorem, one has $\sigma_{\mathbf{e},b|y}=\tr_{\tilde{B}}[(\openone\otimes G_{\mathbf{e},b|y})|\psi_\sigma\rangle\langle\psi_\sigma|]$. On the one hand, we have the sequential form $G_{\mathbf{e},b|y}=\mathcal{I}_{\mathbf{e}}^*(\tilde{B}_{b|y})$. On the other hand, we can use a Stinespring dilation of the instrument to write $\mathcal I_\mathbf{e}(\varrho)=\tr_E[(\openone\otimes E_\mathbf{e}) V\varrho V^*]$. Combining these gives $\sigma_{\mathbf{e},b|y}=\text{tr}_{\tilde{B}E}[(\openone\otimes \tilde B_{b|y}\otimes E_\mathbf{e})\varrho_{A\tilde{B}E}]$.

In a realistic scenario, the recorded statistics for Alice, Bob and Eve read $p(a,b,\mathbf{e}|x,y)=\text{tr}[(A_{a|x}\otimes\tilde B_{b|y}\otimes E_\mathbf{e})\varrho_{A\tilde{B}E}]$. The guessing probability of the canonical scenario is related to this distribution via marginalisation over Alice's outcome. In line with the incompatibility-based randomness certification, we suggest coordinating the guessing rounds and using $p(b,\mathbf{e}|a,x,y)$ instead for randomness. We define the guessing probability as $P_g=\sum_{y,e}p(y)\text{tr}\left[A_{a^*|y+1}\sum_{\{\mathbf{e}|e_{y^*}=e\}} \sigma_{\mathbf{e},e|y}\right]/\text{tr}[\sigma A_{a^*|y+1}]$ with $\sigma=\sum_{\mathbf{e},b}\sigma_{\mathbf{e},b|y}$. This expression can be obtained from Eq.~(\ref{Eq:EveGuess2MUB}) by sandwiching the variables $G_{\mathbf{e},b|y}$ with a square-root of a full-rank state $\sigma$ and normalising the guessing probability.

To bound Eve's knowledge, one can translate the optimisation region of Eq.~(\ref{Eq:EveGuess2MUB}) to the steering setting. Other than the witness bound, the conditions translate directly. We show in the \hyperref[app:E]{Appendix E} that the witness bound is simply replaced by a bound on a steering witness, and that in this way the two optimisation problems can be mapped to one another. In other words, all our randomness bounds based on $G_{\mathbf{e},b|y}$ have a steering-equivalent bound based on $\sigma_{\mathbf{e},b|y}$. This implies specifically a tight link between steerability and randomness certification for any finite number of measurement inputs. As another notable example, one can take a hollow triangle for $\tilde{B}$. For any shared state, the resulting state assemblage $\sum_\mathbf{e}\sigma_{\mathbf{e},b|y}$ is star-unsteerable \cite{minagawa2026} and the effective measurements $\sum_\mathbf{e}G_{\mathbf{e},b|y}$ are star-compatible. It is known that such an assemblage cannot provide randomness in the canonical scenario \cite{minagawa2026}. However, our results, e.g., on the hollow triangle show that randomness can indeed be certified when using the witness-based scenario and a suitably chosen shared state, such as one sufficiently close to a full-Schmidt rank state between Alice and Bob.
 
\section{Conclusions}
We have presented a fully measurement-theoretical analysis of randomness generation in a prepare-and-measure scenario. We have established a trade-off between the capabilities of a classical Eavesdropper and the amount of incompatibility present in the system. This is done by using a known quantifier of incompatibility, the generalised incompatibility robustness, which is shown to be intimately connected to Eve's capabilities and, hence, to certifiable randomness. As the incompatibility robustness can be decided from state discrimination tasks, which are obtainable through a semi-definite program, our work proposes a constructive way for generating randomness from a given set of incompatible quantum measurements, which extends to steerable state assemblages. This is in contrast to existing works, which provide a tight connection between steering and randomness in two input scenarios $\cite{Li2024}$ and in star-steerable scenarios \cite{minagawa2026}. For a given set of measurements, the noise-tolerance of pairwise compatibility and star-compatibility are considerably lower than that of compatibility. Hence, our work can certify randomness from a given set of measurements under noisier environments. Finally, we have used dimensional simulability \cite{Ioannou2022} to show that our framework can be used to certify randomness also in the presence of an Eve with a quantum memory.

Our analysis opens up various possible directions for the near future. A natural direction is to investigate incompatibility monogamy. We have shown that higher values of incompatibility robustness lead to fewer strategies for Eve. As Eve's strategies are limited by the POVMs pairwise compatibility with Bob's effective measurements, we see that there is a monogamy type trade-off. We expect this to be a fruitful direction within our paradigm, i.e. finding optimal strategies for Eve and optimal ways for generating randomness, as well as for foundational studies in measurement theory. We leave the precise analysis of these questions for future works.

We believe that our work can also lead to interesting questions in continuous-variable systems. Many of the measurement-theoretical results, such as the connection between incompatibility robustness and state discrimination tasks, have an infinite-dimensional counterpart \cite{kuramochi20}. This suggests a possible path towards finding counterparts of our results in continuous-variable quantum random number generation scenarios.

One can also consider making the scenario more secure. One direction is to lift the trust assumption on the input states by considering, e.g. dimension-bounded incompatibility \cite{Egelhaaf2025} or self-testing techniques \cite{Supic2020review}. Another direction is to investigate various generalisations of joint measurability, such as joint measurability on many copies \cite{carmeli16} and genuine $n$-wise incompatibility \cite{quintino19}, which may translate to meaningful attacks that a more general Eve could use. Going the other way around, we expect that known attacks of a more general Eve based on, e.g., simultaneous measurements on many particles, can find intriguing measurement-theoretical counterparts.

\section*{Acknowledgements} 

We thank Sophie Egelhaaf, Jef Pauwels and Pavel Sekatski for very fruitful discussions. The work was supported by the Swedish Research Council
(grant no. 2024-05341) and the Wallenberg Initiative on Networks and Quantum Information (WINQ).
MP is thankful to the support from the Niedersächsisches Ministerium für Wissenschaft und Kultur. LL acknowledges financial support from the Business Finland project BEQAH.

\textit{Note added.---} While finalising this manuscript, we became aware of a recent related work \cite{lobo2026generalizedmeasurementincompatibility}, which presents a qualitative connection between measurement incompatibility and Eve's guessing probability. The link requires Eve to guess the outcome for a collection of input states, but more general notions of joint measurability are considered, while our method only requires Eve to guess in the case of some full-rank state, which is unknown to her.

\bibliography{References}

\clearpage
\onecolumngrid
\newgeometry{top=1in, bottom=1in, left=1.5in, right=1.5in}
\appendix

\section{Sequential implementation of the pairwise joint measurements of Eve and Bob.}\label{app:A}

We include here for completeness the instruments and measurements needed for the identity $G_{\lambda,b|y}=\sqrt{E_\lambda}\cE^*_\lambda(\tilde B_{b|y})\sqrt{E_\lambda}$ [Eq.~\eqref{eq:Glambdaby} in the main text] of the pairwise joint measurements $\{G_{\lambda, b|y}\}$ of $\{B_{b|y}\}$ and $\{E_\lambda\}$.

Let $(\cH_0,\{P_\lambda\},J)$ be a minimal Naimark dilation of $\{E_\lambda\}$. We then define the following instrument for Eve.
\begin{align}
    \mathcal{I}_\lambda^*(\cdot):=J^*P_\lambda \cdot P_\lambda J
\end{align}

This instrument clearly implements Eve's POVM. Thus it can also be given in the form $\mathcal{I}_\lambda^*(\cdot)=\sqrt{E_\lambda}\cE^*_\lambda( \cdot) \sqrt{E_\lambda}$, where each $\cE_\lambda$ is a quantum channel \cite[Corollary 1]{Pello4}. Let then $G_{\lambda,b|y}$ be the pairwise joint measurement. Since for every $b,y$ we have $G_{\lambda,b|y}\leq E_\lambda$,  by e.g. Theorem 1.4.2 in \cite{arveson69} or Lemma 1 in \cite{haapasalo15} (see also Lemma 1 in \cite{pellonpaa14}) there exists a $0\leq \tilde{B}_{b|y}\leq \openone$ for every $b,y$ such that $[P_\lambda,\tilde{B}_{b|y}]=0$ for all $\lambda$, and $G_{\lambda,b|y}=J^*P_\lambda \tilde{B}_{b|y}J$. Furthermore as the dilation is minimal, for all $y$ the identity $J^*P_\lambda \sum_b\tilde{B}_{b|y}J=E_\lambda=J^*P_\lambda J$ implies that  $\sum_b  \tilde{B}_{b|y}=\openone$, i.e.  the collection $\{\tilde{B}_{b|y}\}$ is a POVM. Thus the equality $\sqrt{E_\lambda}\cE_\lambda^*(\tilde B_{b|y}) \sqrt{E_\lambda}=J^*P_\lambda \tilde{B}_{b|y}J=G_{\lambda,b|y}$ finishes the proof. 

\section{Details of the relation between randomness and incompatibility.}\label{app:B}

In this Appendix we provide details of the relation between measurement incompatibility and randomness for full rank witnesses. We first focus on the case where Eve tries to guess all inputs $y$ and later consider the more general case of subsets of inputs.

Let us focus on the direction where Bob's measurements $\{B_{b|y}\}$ are jointly measurable. Let $n=\#y$ be the number of inputs of Bob, and $\mathbf{b}=(b_1,\dots,b_n)$ be a vector consisting of outcomes. Now by joint measurability there is a marginal joint measurement $\{G_{\mathbf{b}}\}$ such that for all $y$ and $b$ we have the following.
\begin{align}
   B_{b|y}=\sum_{\{\mathbf{b}|b_y=b\}} G_{\mathbf{b}}.
\end{align}
We can then take a (minimal) Naimark dilation of $\{G_{\mathbf{b}}\}$, $G_{\mathbf{b}}=J^*P_{\mathbf{b}}J$ and define Eve's instrument in the Heisenberg picture as follows. Here $A$ is some bounded operator.
\begin{align}
    I_{\mathbf{e}}^*(A):=J^*P_{\mathbf{e}}AP_{\mathbf{e}}J
\end{align}
Bob's black-box measurements are then the marginals $P_{b|y}=\sum_{\{\mathbf{b}|b_y=b\}} P_{\mathbf{b}}$. Indeed one easily confirms that $\sum_{\mathbf{b}}I_{\mathbf{b}}^*(P_{b|y})=B_{b|y}$. 

Finally, suppose randomness is generated using the input $y^*$. Then Eve can guess perfectly by doing a marginal post-processing. Indeed, we see the following. 
\begin{align}
    \sum_{\{\mathbf{e}|e_{y^*}=e\}}I_{\mathbf{e}}^*(P_{b|y^*})=J^*P_{e|y^*}P_{b|y^*}J=\delta_{be}B_{b|y^*}
\end{align}
Especially $\sum_{\{\mathbf{e}|e_{y^*}=e\}}I_{\mathbf{e}}^*(P_{e|y^*})=B_{e|y^*}$ so that for any state $\rho_{b'|y'}$ Alice sends, we have the following.

$$p_{E|B}(e|e,b',y',y^*)=\frac{p_{EB}(e,e|b',y',y^*)}{p_B(e|b',y',y^*)}=\frac{\text{tr}\left[\sum_{\{\mathbf{e}|e_{y^*}=e\}}I_{\mathbf{e}}^*(P_{e|y^*})\varrho_{b'|y'}\right]}{\text{tr}\left[B_{e|y^*}\varrho_{b'|y'}\right]}=1.$$

Thus for every $y^*$, Eve can post-process her POVM to achieve perfect guessing probability regardless of the states sent in.

Let us now focus on the other direction. In other words, we assume that Eve has perfect guessing probability and that Alice is sending full rank states in the randomness rounds. Explicitly there exists a strategy for Eve resulting in the pairwise joint distribution $\{G_{\mathbf{e},e|y}\}$ fulfilling conditions of the SDP \eqref{Eq:EveGuess2MUB} such that the following holds.
\begin{align}\label{Eq:perfectguessing1}
    \sum_e\sum_y p(y)\text{tr}\left[\varrho_{b'|y+1}\sum_{\{\mathbf{e}|e_y=e\}} G_{\mathbf{e},e|y}\right]=1
\end{align}

Since $\sum_{e,\mathbf{e}}G_{\mathbf{e},e|y}=\openone$, we can rewrite the above Equation in the following form.
\begin{align}
\sum_e\sum_y p(y)\text{tr}\left[\varrho_{b'|y+1}\left(\sum_\mathbf{e} G_{\mathbf{e},e|y}-\sum_{\{\mathbf{e}|e_y=e\}}G_{\mathbf{e},e|y}\right)\right]=0
\end{align}
We assume here that $p(y)>0$ for all $y$ as Eve tries to guess perfectly for all inputs simultaneously. Clearly $\sum_\mathbf{e} G_{\mathbf{e},e|y}-\sum_{\{\mathbf{e}|e_y=e\}}G_{\mathbf{e},e|y}\ge 0$, so that the faithfulness of the $\rho_{b'|y+1}$ implies the following for all $e$ and for all $y$.
\begin{align}
\sum_\mathbf{e} G_{\mathbf{e},e|y}=\sum_{\{\mathbf{e}|e_y=e\}}G_{\mathbf{e},e|y}
\end{align}
This then further implies that $\sum_{\{\mathbf{e}|e_y=b\}} G_{\mathbf{e},e|y}=0$ if $b\neq e$. Thus we compute the following.
\begin{align}\label{Eq:perfectguessing2}
    \sum_\mathbf{e} G_{\mathbf{e},e|y}=\sum_{\{\mathbf{e}|e_y=e\}}G_{\mathbf{e},e|y}=\sum_{\{\mathbf{e}|e_y=e\},e'}G_{\mathbf{e},e'|y}=\sum_{\{\mathbf{e}|e_y=e\}}E_{\mathbf{e}}
\end{align}
In the last step we used the fact that the marginal over Bob does not depend on the input. In other words, for all inputs $y,y'$ we have $\sum_{e'}G_{\mathbf{e},e'|y}=\sum_{e'}G_{\mathbf{e},e'|y'}=:E_{\mathbf{e}}$.
Thus Bob's effective measurements $\left\{\sum_\mathbf{e} G_{\mathbf{e},e|y}\right\}$ are all jointly measurable.

Next we consider the case where Eve can not necessarily guess everything for every input of Bob. We now show how incompatibility limits Eve in this case. 

Consider a subset $S$ of Bob's inputs. We now show, assuming that Alice sends a faithful incompatibility witness, that Eve can guess perfectly for every input in $S$ if and only if for every input $y'$ the set $\{B_{b|y}\}_{y \in S} \cup \{B_{b|y'}\}$ of Bob's effective measurements
is jointly measurable. Especially the case, where for every $S$ with $\#S=k$ Eve has a strategy to guess perfectly for every input in $S$, is equivalent to every $k+1$-sized subset of Bob's effective measurements being jointly measurable. 

We choose the objective function of the SDP \eqref{Eq:EveGuess2MUB} to suit the situation explained above. Explicitly, given $S$, we define the objective function as follows. 
\begin{align}\label{Eq:guessingsubset}
    \sum_e\sum_{y\in S} p(y)\text{tr}\left[\varrho_{b'|y+1}\sum_{\{\mathbf{e}|e_y=e\}} G_{\mathbf{e},e|y}\right]
\end{align}
Note that now Eve is guessing only inputs in $S$ so that $\mathbf{e}=(e_1,\dots,e_k)$ with $k=\#S$. Thus in this case $p(y)=0$ for $y \notin S$. 

Let us first show that if $\{B_{b|y}\}_{y \in S} \cup \{B_{b|y'}\}$ is jointly measurable for every input, then we achieve probability 1 in Equation \eqref{Eq:guessingsubset}. For this, let $y'$ be an arbitrary input.
Now since $\{B_{b|y}\}_{y \in S}$ is jointly measurable, there is a marginal-form joint measurement $\{G_{\mathbf{b|}S}\}$. Let us now define $G_{\mathbf{e},b|y'}:=G_{\mathbf{e}b|S,y'}$ as a pairwise joint measurement of $\{G_{\mathbf{e|}S}\}$ and $\{B_{b|y'}\}$, which exists by the assumption that $\{B_{b|y}\}_{y \in S} \cup \{B_{b|y'}\}$ is jointly measurable. If $y' \in S$, we use the explicit pairwise joint measurement $G_{\mathbf{e}b|S,y'}=\delta_{e_{y'}b}G_{\mathbf{e}}$. Note that as mentioned in the main text, each of these can be implemented with an instrument and a black-box measurement: $G_{\mathbf{e},b|y'}=\sqrt{G_\mathbf{e}}\mathcal{E}_{\mathbf{e}}^*(\tilde{B}_{b|y'})\sqrt{G_{\mathbf{e}}}$. Clearly $\{G_{\mathbf{e},b|y'}\}$ fulfill the conditions of the SDP \eqref{Eq:EveGuess2MUB} (adapted to this case). Let us the finally show that unit probability is reached with this in Eq. \eqref{Eq:guessingsubset}.
\begin{align}
    \sum_e\sum_{y\in S} p(y)\text{tr}\left[\varrho_{b'|y+1}\sum_{\{\mathbf{e}|e_y=e\}} G_{\mathbf{e},e|y}\right]&=\sum_e\sum_{y\in S} p(y)\text{tr}\left[\varrho_{b'|y+1}\sum_\mathbf{b} \delta_{b_ye}\delta_{b_ye}G_{\mathbf{e|}S}\right] \\
    &=\sum_e\sum_{y\in S} p(y)\text{tr}\left[\varrho_{b'|y+1}B_{e|y}\right]=1
\end{align}
We then concentrate on the other direction. Suppose that 
\begin{align}
    \sum_e\sum_{y\in S} p(y)\text{tr}\left[\varrho_{b'|y+1}\sum_{\{\mathbf{e}|e_y=e\}} G_{\mathbf{e},e|y}\right]=1
\end{align}
Again, we assume that the $\rho_{b'|y+1}$ are of full rank and that $p(y)>0$ for all $y$. Then repeating the steps in Equations \eqref{Eq:perfectguessing1}-\eqref{Eq:perfectguessing2} we see that Bob's effective measurements $\left\{\sum_{\mathbf{e}} G_{\mathbf{e},e|y}\right\}$ are jointly measurable for all $y \in S$ or in other words $\sum_{\mathbf{e}}G_{\mathbf{e},e|y}=\sum_{\{\mathbf{e}|e_y=e\}} E_\mathbf{e}$, where $E_\mathbf{e}:=\sum_e G_{\mathbf{e},e|y}$. Let now $y'$ be any input. Then $\{G_{\mathbf{b},b'|y'}\}$ is clearly a marginal form joint measurement for the set $\left\{\sum_{\mathbf{e}} G_{\mathbf{e},b|y}\right\}_{y \in S} \cup \{\sum_{\mathbf{e}} G_{\mathbf{e},b'|y'}\}$, thus finishing the proof.

\section{Perfect guessing with qubit memory implies qubit simulability} \label{app:memory-toSimulability}
Suppose again that the pairwise joint distributions of Eve and effective Bob are denoted $G_{e,b|y}$. Perfect guessing now implies that $\sum_{y,e}p(y)\text{tr}\left[\varrho_{b^*|y+1} G_{e,e|y}\right]=1$.
Thus especially if we fix an input $y^*$, we have that $\sum_{e}\text{tr}\left[\varrho_{b^*|y^*+1} G_{e,e|y^*}\right]=1=\sum_e \tr{\varrho_{b^*|y^*+1} B_{e|y^*}}$ for $p(y^*)\neq0$, where $\sum_b G_{e,b|y^*}$ qubit simulable. Then as $G_{e,e|y^*}\leq B_{e|y^*}$ we see that $B_{e|y^*}=G_{e,e|y^*}$. Therefore $$B_{e|y^*}= G_{e,e|y^*}= \sum_b G_{e,b|y^*}$$ Thus $\{B_{e|y^*}\}$ is qubit simulable for all $y^*$ such that $p(y^*)\neq 0$. Especially if $p(y)>0$ for all $y$, then $\{B_{b|y}\}$ is qubit simulable for all $y$.

\section{Qubit simulability from qubit memory} \label{app:memory-fromSimulability}
Suppose Eve has a protocol where she access to a qubit memory (or some higher dimensional memory as all the following proofs stay the same for higher dimensions). 

As mentioned in the main text, the pairwise joint distributions of Eve and Bob can then be given as follows.
\begin{align}
    G_{e,b|y}=\sum_{\lambda,\lambda'}p(e|\lambda,\lambda') \mathcal I_\lambda^*(M_{\lambda'|y,\lambda}\otimes\tilde{B}_{b|y})
\end{align}
Here $M_{\lambda'|y,\lambda}$ is Eve's POVM on the qubit memory.
Based on this we define the following.

\begin{align}
    \mathcal{J}_\lambda^*&:=\mathcal{I}_\lambda^*(\cdot \otimes I)\\
    \Lambda&:=\sum_\lambda \mathcal{J}_\lambda \otimes \ket{\lambda}\bra{\lambda} \\
    M_{e|y}&:=\sum_{\lambda,\lambda'} p(e|\lambda,\lambda') M_{\lambda'|y,\lambda} \otimes \ket{\lambda}\bra{\lambda}
\end{align}
Here $\{\ket{\lambda}\}$ is some orthonormal basis "flag". With a simple calculation one sees that $\sum_bG_{e,b|y}=\Lambda^*(M_{e|y})$. Furthermore since $\{\mathcal{J_\lambda}\}$ is an instrument mapping to a qubit, all $J_\lambda$ have maximum Kraus rank 2. Therefore also $\Lambda$ has maximum Kraus rank 2 i.e. $\Lambda$ is a Schmidt-rank 2 channel. Thus $\left\{\sum_b G_{e,b|y}\right\}$ can be prepared with a Schmidt rank 2 channel and is therefore qubit simulable \cite{Ioannou2022}.
\section{Mapping between the steering scenario and the incompatibility scenario}\label{app:E}

In this Appendix we show that the SDP \eqref{Eq:EveGuess2MUB} is equivalent to the following steering-based SDP. Here $\sigma$ is some fixed full rank state.
\begin{align}\label{Eq:SteeringSDP}
        P_g(\sigma) =  \max &\sum_{y,e}p(y)\frac{\text{tr}\left[A_{b^*|y+1}\sum_{\{\mathbf{e}|e_{y}=e\}} \sigma_{\mathbf{e},e|y}\right]}{\text{tr}[\sigma A_{b^*|y+1}]} \\
        \text{s.t. } &\sum_{b,y}\text{tr}\left[ A_{b|y}\sum_{\mathbf{e}} \sigma_{\mathbf{e},b|y}\right]\geq (1+\alpha)\beta_{\text{CSR}}\nonumber\\
        &\sigma_{\mathbf{e},b|y}\geq0, \;\;\sum_{\mathbf{e},b}\sigma_{\mathbf{e},b|y}=\sigma\nonumber\nonumber\\
        &\sum_b \sigma_{\mathbf{e},b|y}=\sum_b\sigma_{\mathbf{e},b|y'}.\nonumber
\end{align}
Here $\beta_{\text{CSR}}$ refers to the classical bound with respect to the steering witness $A_{b|y}$, related to the consistent steering robustness \cite{cavalcanti16}. Let us show that this SDP can be mapped to \eqref{Eq:EveGuess2MUB} and vice versa.

Suppose that we are given the SDP \eqref{Eq:SteeringSDP}. Then defining $G_{\mathbf{e},b|y}:=\sigma^{-1/2}\sigma_{\mathbf{e},b|y}\sigma^{-1/2}$ and $\varrho_{b|y}:=\sigma^{1/2}A_{b|y}\sigma^{1/2}$ and plugging these into \eqref{Eq:SteeringSDP}, we see the following.
\begin{align}
        P_g(\sigma) =  \max &\sum_{y,e}p(y)\text{tr}\left[\frac{\varrho_{b^*|y+1}}{\text{tr}[\varrho_{b^*|y+1}]}\sum_{\{\mathbf{e}|e_{y}=e\}} G_{\mathbf{e},e|y}\right] \\
        \text{s.t. } &\sum_{b,y}\text{tr}\left[ \varrho_{b|y}\sum_{\mathbf{e}} G_{\mathbf{e},b|y}\right]\geq (1+\alpha)\beta_{\text{CSR}}\nonumber\\
        &G_{\mathbf{e},b|y}\geq0, \;\;\sum_{\mathbf{e},b}G_{\mathbf{e},b|y}=\openone\nonumber\nonumber\\
        &\sum_b G_{\mathbf{e},b|y}=\sum_b G_{\mathbf{e},b|y'}.\nonumber
\end{align}
This thus coincides with \eqref{Eq:EveGuess2MUB}, if for the incompatibility witness $\varrho_{b|y}=\sigma^{1/2}A_{b|y}\sigma^{1/2}$ we have $\beta_{\text{JM}}=\beta_{\text{CSR}}$. But this is clear as the following shows. Here we denote the set of unsteerable assemblages by $\text{US}$.
\begin{align}
\beta_{\text{CSR}}&=\sup \left\{ \sum_{b,y}\text{tr}\left[ A_{b|y} \sigma_{b|y}\right] \, \bigg| \, \{\sigma_{b|y}\} \in \text{US}, \, \sum_b \sigma_{b|y}=\sum_{\mathbf{e},b}\sigma_{\mathbf{e},b|y}=\sigma \right\}
\\
&=\sup \left\{ \sum_{b,y}\text{tr}\left[ \varrho_{b|y}(\sigma^{-1/2}\sigma_{b|y}\sigma^{-1/2}) \right] \, \bigg| \,\{\sigma_{b|y}\} \in \text{US}, \, \sum_b \sigma_{b|y}=\sigma \right\} \\
&=\sup \left\{ \sum_{b,y}\text{tr}\left[ \varrho_{b|y}M_{b|y} \right] \, \bigg| \,\{M_{b|y}\} \in \text{JM} \right\}=\beta_{\text{JM}}
\end{align}
Here $\sigma$ was an arbitrary full rank state and thus $P_g(\sigma)$ coincides with \eqref{Eq:EveGuess2MUB} for all full rank $\sigma$ with the related witness. 

For the converse direction, for any full rank state $\sigma$, one defines $\sigma_{\mathbf{e},b|y}:=\sigma^{1/2}G_{\mathbf{e},b|y}\sigma^{1/2}$ and $A_{b|y}:=\sigma^{-1/2}\varrho_{b|y}\sigma^{-1/2}$. Now $\text{tr}[\sigma A_{b^*|y+1}]=\text{tr}[\varrho_{b|y}]=1$, and furthermore $\beta_{\text{CSR}}=\beta_{\text{JM}}$ by the previous calculation. Thus plugging these definitions into the SDP \eqref{Eq:EveGuess2MUB} we easily see that it coincides with $P_g(\sigma)$.

\end{document}